\documentclass[12pt]{article}
\usepackage{amsmath}
\usepackage{graphicx,psfrag,epsf}
\usepackage{enumerate}
\usepackage{psfrag,epsf}
\usepackage{url} 
\usepackage{amsmath,amssymb,amsfonts,amsthm,mathtools,mathrsfs}

\usepackage{bm,bbm}
\usepackage{color}
\usepackage{subfigure}
\usepackage{algorithm,algpseudocode}
\usepackage[symbol]{footmisc}
\usepackage{scalefnt}
\usepackage{authblk} 
\usepackage{multirow,centernot}
\usepackage[colorlinks=true, citecolor=blue, urlcolor=blue]{hyperref}
\usepackage{natbib}
\usepackage{dsfont}
\usepackage[title]{appendix}
\usepackage{newtxtext,newtxmath} 
\usepackage{mhchem} 

\input cyracc.def

\usepackage[left=1in,top=1.1in,right=0.5in,bottom=1in]{geometry}

\theoremstyle{definition}

\newtheorem{theorem}{Theorem}[section]

\newtheorem{remark}[theorem]{Remark}

\makeatletter
\def\@seccntformat#1{\@ifundefined{#1@cntformat}%
	{\csname the#1\endcsname\quad}
	{\csname #1@cntformat\endcsname}
}
\makeatother

\newif\ifShowComments
\ShowCommentstrue
\def\strutdepth{\dp\strutbox}
\def\druk#1{\strut\vadjust{\kern-\strutdepth
        {\vtop to \strutdepth{%
                \baselineskip\strutdepth\vss
                        \llap{\hbox{#1}\quad}\null}}}}

\title{\bf
Novel closed-form point estimators for a weighted exponential family derived from likelihood equations
}

\author[1]{  Roberto Vila \thanks{rovig161@gmail.com}}
\author[1]{Eduardo Nakano  \thanks{eynakano@gmail.com} }
\author[1]{Helton Saulo  \thanks{heltonsaulo@gmail.com} }
\affil[1]{Department of Statistics, University of
	Bras\'ilia, 70910-900, Bras\'ilia, Brazil}
\begin{document}
	\maketitle 	
	\begin{abstract}
		{
In this paper, we propose and investigate closed-form point estimators for a weighted exponential family. We also develop a bias-reduced version of these proposed closed-form estimators through bootstrap methods. Estimators are assessed using a Monte Carlo simulation, revealing favorable results for the proposed bootstrap bias-reduced estimators.
}

	\end{abstract}
	
	
	\section{Introduction}
	\noindent


%
The aim of this work is to provide closed-form estimators for the parameters of a parametric set of probability distributions that belongs to the weighted exponential family of the following form
%
%
%
\begin{align}\label{pdf-1}
f(x;\psi)
=
{(\mu\sigma)^{\mu+1} \over (\sigma+\delta_{ab})\Gamma(\mu+1)}\,
[1+\delta_{a,b}T(x)]\,
{\vert T'(x)\vert\over T(x)}\,
\exp\left\{-\mu \sigma T(x)+\mu\log(T(x))\right\},
\quad 
x\in D\subseteqq \mathbb{R},
\end{align}
where $\psi=(\mu,\sigma)$, $\mu,\sigma>0$, $T:D\to (0,\infty)$ is a
real strictly monotone twice differentiable function and $\delta_{ab}$ is the Kronecker delta function, that is, $\delta_{ab}$ is 1 if $a=b$, and 0 otherwise.
In the above, $T'(x)$ denotes the derivative of $T(x)$ with respect to $x$.

The probability function $f(x;\psi)$ in \eqref{pdf-1} can be interpreted as a mixture of two distributions that belongs to the exponential family, that is,
\begin{align}\label{decomp-1}
	f(x;\psi)&=
	{\sigma\over \sigma+\delta_{ab}}\,
	f_1(x)
	+
	{\delta_{ab}\over \sigma+\delta_{ab}}\,
	f_2(x),
	\quad x\in D\subseteqq \mathbb{R}, \ \mu, \sigma>0,
	\\[0,3cm]
	&=
		\begin{cases} 
		\displaystyle
		%
			{\sigma\over \sigma+1}\,
		f_1(x)
		+
		{1\over \sigma+1}\,
		f_2(x),
		& a=b,
		\\[0,55cm]
		\displaystyle
		f_1(x),
		& a\neq b,
	\end{cases} \nonumber
\end{align}
where
\begin{align}\label{def-fi}
	f_j(x)
	=
	{(\mu\sigma)^{\mu+j-1} \over \Gamma(\mu+j-1)}\,
	{\vert T'(x)\vert\over T(x)}\,
	\exp\left\{-\mu \sigma T(x)+(\mu+j-1)\log(T(x))\right\},
	\quad 
	x\in D, \quad j=1,2.
\end{align}
Probability densities of form $f_j(x)$, $j=1,2$, have appeared in \cite{Nascimento2014} and more recently in \cite{Vila2024}.

If $X$ has density in \eqref{pdf-1}, from \eqref{decomp-1} and \eqref{def-fi} it is simple to observe that
\begin{align}\label{rep-st-1}
	X\stackrel{\mathscr{D}}{=}(1-B)T^{-1}(Z_1)+BT^{-1}(Z_2),
\end{align}
where ``$\stackrel{\mathscr{D}}{=}$'' means equality in distribution, $B\in\{0,1\}$ is a Bernoulli random variable with success parameter $\delta_{ab}/(\sigma+\delta_{ab})$, independent of $Z_1$ and $Z_2$, such that $Z_j\sim{\rm Gamma}(\mu+j-1,1/(\mu\sigma))$, $j=1,2$. Here, $T^{-1}$ denotes the inverse function of $T$.

Since $T$ is arbitrary and
real strictly monotone twice differentiable function, our model covers a wide range of distributions including the models studied in \cite{Kim2022}. The model \eqref{pdf-1} is reduced to those distributions mentioned in \cite{Kim2022} by relating the transformation $T$ with the extension of the Box-Cox transformation as follows: $T^{-1}(x)=\log(x^{1/\gamma})$, $x>0$, and $a\neq b$.



Table \ref{table:1-0} presents some examples of generators $T(x)$ for use in \eqref{pdf-1} with $a=b$.

\begin{table}[H]
	\caption{Some examples of 
		generators $T(x)$ of probability densities \eqref{pdf-1} with $D=(0,\infty)$ and $a= b$.}
	\vspace*{0.15cm}
	\centering 
		\resizebox{\linewidth}{!}{
			\begin{tabular}{lcccccccl} 
				\hline
				Distribution & $\mu$ & $\sigma$ & $T(x)$&$T'(x)$&$T''(x)$ & Parameters
				\\ [0.5ex] 
				\noalign{\hrule height 1.0pt}
				\\ [-1.5ex] 
				Weighted Lindley \citep{Kim2021}
				&  $\phi$  & ${\lambda\over \phi}$   & $x$&1& 0&  $\lambda,\phi>0$		
				\\ [2.0ex] 
Weighted inverse  Lindley  &  $\phi$  & ${\lambda\over \phi}$   & ${1\over x}$&$-{1\over x^2}$& ${2\over x^3}$&  $\lambda,\phi>0$
				\\ [2.0ex] 
New weighted exponentiated Lindley  &  $\phi$  & ${\lambda\over \phi}$   & ${\log(x+1)}$&${1\over x+1}$& $-{1\over (x+1)^2}$&  $\lambda,\phi>0$
				\\ [2.0ex] 
New weighted log-Lindley  &  $\phi$  & ${\lambda\over \phi}$   & $\exp(x)-1$&$\exp(x)$& $\exp(x)$&  $\lambda,\phi>0$
		\\ [2.0ex]  
		Weighted Nakagami 
		&  $m$  & ${1\over \Omega}$   & $x^2$&$2x$& $2$&  $m\geqslant {1\over 2}, \ \Omega>0$		
		\\ [2.0ex] 
		Weighted inverse  Nakagami  &   $m$  & ${1\over \Omega}$    & ${1\over x^2}$&$-{2\over x^3}$&${6\over x^4}$&  $m\geqslant {1\over 2}, \ \Omega>0$
		\\ [2.0ex] 
		New weighted exponentiated Nakagami  &   $m$  & ${1\over \Omega}$   & ${\log(x^2+1)}$&${2x\over x^2+1}$& 
		$-{2 ( x^2-1)\over(x^2+1)^2}$&  $m\geqslant {1\over 2}, \ \Omega>0$
		\\ [2.0ex] 
		New weighted log-Nakagami  &   $m$  & ${1\over \Omega}$  & $\exp(x^2)-1$&$2x\exp(x^2)$&$2\exp(x^2) (1 + 2 x^2)$ &  $m\geqslant {1\over 2}, \ \Omega>0$
		\\ [1.0ex] 
		\hline	
\end{tabular}
}
\label{table:1-0} 
\end{table}

For  $a\neq b$,
Table 1 of \cite{Vila2024} provides some examples of generators $T(x)$, with $D=(0,\infty)$, for use in \eqref{pdf-1}.

\bigskip 
There are various methodological proposals in the literature for obtaining close-form estimators. The methodologies for obtaining the close-form estimators can be moment-based-type \citep{Cheng-Beaulieu2002} or based on score-adjusted approach \citep{Nawa2023, Tamae2020}, for example. There are also proposals to obtain close-form estimators based on the likelihood function. One type of likelihood-based estimator is obtained by considering the likelihood equations of the generalized distribution obtained through a power transformation and taking the baseline distribution as a special case. Several works adopt this idea, sometimes with slight variations in the power transformation. See for example \cite{YCh2016}, \cite{Vila2024}, \cite{RLR2016}, \cite{Kim2021} and \cite{Zhao2021}. In addition, other proposals for likelihood-based estimators can also be seen in \cite{Kim2022}, that adopt an extension of the Box-Cox transformation for obtaining a closed-form estimator and \cite{Cheng-Beaulieu2001}, who proposed a closed-form estimator by an asymptotic expansion.

This work develops a likelihood-based closed-form estimator for the parameters of probability distributions of the weighted exponential family \eqref{pdf-1} using the generalizing distribution procedure through a power transformation.

The rest of the paper proceeds as follows. In Section \ref{A generalization of weighted exponential family}, we state briefly a generalization of weighted exponential
family. In Section \ref{The New Estimators} and \ref{Asymptotic behavior of estimators} we describe the new proposed estimation methods and some asymptotic results. Some closed-form estimators
of the parameters of distributions belonging to the weighted exponential class are presented in Section \ref{Sec:5}. Furthermore, in Section \ref{Sec:6}, we perform a Monte Carlo simulation study to assess the
performance of a bootstrap bias-reduced version of these proposed closed-form estimators.



\section{A generalization of weighted exponential family}\label{A generalization of weighted exponential family}
By using the increasing power transformation $Y=X^{1/p}$, $p>0$, considered in \cite{YCh2016} and \cite{Cheng-Beaulieu2002}, where $X$ has the PDF in \eqref{pdf-1}, it is clear that the PDF of $Y$ is written as
\begin{align}\label{dist-gen-exp}
f(y;\psi,p)
=
p\, 
{(\mu\sigma)^{\mu+1} \over (\sigma+\delta_{ab})\Gamma(\mu+1)}\,
[1+\delta_{ab}T(y^p)]\,
{\vert T'(y^p)\vert y^{p-1}\over T(y^p)}\,
\exp\left\{-\mu \sigma T(y^p)+\mu\log(T(y^p))\right\},
\end{align}
where $y\in D\subseteqq \mathbb{R}$, $\psi=(\mu,\sigma)$ and $\mu,\sigma, p>0$.

From \eqref{rep-st-1} it follows that
\begin{align}\label{rep-st-2}
Y\stackrel{\mathscr{D}}{=}(1-B)[T^{-1}(Z_1)]^{1/p}+B[T^{-1}(Z_2)]^{1/p},
\end{align}
where $B\sim {\rm Bernoulli}(\delta_{ab}/(\sigma+\delta_{ab}))$  is independent of $Z_1$ and $Z_2$, such that $Z_j\sim{\rm Gamma}(\mu+j-1,1/(\mu\sigma))$, $j=1,2$.

\section{The new estimators}\label{The New Estimators}

Let $\{Y_i : i = 1,\ldots , n\}$ be a univariate random sample of size $n$ from  $Y$.
%
The (random) likelihood function for $(\psi,p)$ is written as
\begin{align*}
L(\psi,p)    
&=
p^n\, {(\mu\sigma)^{n(\mu+1)}  \over (\sigma+\delta_{ab})^n\Gamma^n(\mu+1)}\,
\prod_{i=1}^{n}
[1+\delta_{ab}T(Y_i^p)]
\prod_{i=1}^{n}
{\vert T'(Y_i^p)\vert Y_i^{p-1}\over T(Y_i^p)}
\\[0,2cm]
&
\times
\exp\left\{-\mu \sigma \sum_{i=1}^{n} T(Y^p_i)+\mu\sum_{i=1}^{n}\log(T(Y^p_i))\right\}.  
\end{align*}
Consequently, the (random) log-likelihood function for $(\psi,p)$ is written as
\begin{align*}
\log(L(\psi,p))    
&=
n\log(p)
+
{n\mu}\log(\mu) 
+ 
{{n(\mu+1)} \log(\sigma)
	-
	n\log(\sigma+\delta_{ab})
%
-n\log(\Gamma(\mu))}
\\[0,2cm]
&
+
\sum_{i=1}^{n}
\log[1+\delta_{ab}T(Y_i^p)]
+
\sum_{i=1}^{n}
\log(\vert T'(Y_i^p)\vert)
+
(p-1)
\sum_{i=1}^{n}
\log(Y_i)
\\[0,2cm]
&-
\mu \sigma \sum_{i=1}^{n} T(Y^p_i)
+
(\mu-1)\sum_{i=1}^{n}\log(T(Y^p_i)).  
\end{align*}

A simple calculus shows that the elements of the (random) score vector are given by
\begin{align*}
 {\partial \log(L(\psi,p)) \over\partial \mu}
 &=
 n+n\log(\mu)+n\log(\sigma)-n\psi^{(0)}(\mu)
-
\sigma \sum_{i=1}^{n} T(Y^p_i)
+
\sum_{i=1}^{n}\log(T(Y^p_i)),  
 \\[0,2cm]
  {\partial \log(L(\psi,p)) \over\partial \sigma}
 &=
  {n(\mu+1) \over \sigma} 
  -
  {n\over \sigma+\delta_{ab}}
  -
  \mu \sum_{i=1}^{n} T(Y^p_i),
 \\[0,2cm]
  {\partial  \log(L(\psi,p)) \over\partial p}
 &=
 {n\over p}
 +
 \delta_{ab}
  \sum_{i=1}^{n}
  {T'(Y_i^p)Y_i^p\log(Y_i)\over 1+\delta_{ab}T(Y_i^p)}
 +
 \sum_{i=1}^{n}
{T''(Y^p_i)\over T'(Y_i^p)}\, Y^p_i \log(Y_i)
+
\sum_{i=1}^{n}
\log(Y_i)
\\[0,2cm]
&
-
\mu \sigma \sum_{i=1}^{n} T'(Y^p_i) Y^p_i \log(Y_i)
+
(\mu-1)\sum_{i=1}^{n} {T'(Y^p_i)\over T(Y^p_i)}\,  Y^p_i \log(Y_i), 
\end{align*}
where $\psi^{(m)}(x)=\partial^{m+1} \log(\Gamma(x))/\partial x^{m+1}$ is the polygamma function of order $m$. Setting these equal to zero and solving the system of equations gives the maximum likelihood (ML) estimators of $(\psi, p)$.

\subsection{Case $a=b$}\label{Case equal}
By resolving  ${\partial \log(L(\psi,p))/\partial p}=0$, we can express $\mu$ as a function of $\sigma$ and $p$ as follows:
%
\begin{align}\label{estimator-mu}
	{\mu}(\sigma,p)
=
%
{\color{black}
	\dfrac{\overline{Z}(p)}{\sigma\overline{Y}_5(p)-\overline{Y}_4(p)},
}
\end{align}
where we adopt the following notations:
\begin{align*}
&\overline{Z}(p) 
\equiv
{1\over p}
+
\overline{Y}_1(p)
+
\overline{Y}_2(p)
+
\overline{Y}_3(p)
-
\overline{Y}_4(p),
\\[0,2cm]
&\overline{Y}_1(p)
\equiv
{1\over n}
\sum_{i=1}^{n}
\log(Y_i),
\\[0,2cm]
&\overline{Y}_2(p)
\equiv
{1\over n}
\sum_{i=1}^{n}
{T''(Y^p_i)\over T'(Y_i^p)}\, Y^p_i \log(Y_i),
\\[0,2cm]
&\overline{Y}_3(p)
\equiv
{1\over n}
\sum_{i=1}^{n}
{T'(Y_i^p)\over 1+T(Y_i^p)}\, Y_i^p\log(Y_i),
\\[0,2cm]
&\overline{Y}_4(p)
\equiv
{1\over n}\sum_{i=1}^{n} {T'(Y^p_i)\over T(Y^p_i)}\,  Y^p_i \log(Y_i),
\\[0,2cm]
&\overline{Y}_5(p)
\equiv
{1\over n} \sum_{i=1}^{n} T'(Y^p_i) Y^p_i \log(Y_i).
\end{align*}

By substituting \eqref{estimator-mu} into ${\partial \log(L(\psi,p))/\partial \sigma}=0$, we
obtain the closed-form estimator for $\sigma$ as follows:
\begin{align}\label{sigma-estimator}
	\widehat{\sigma}(p)
	=
	{\color{black}
	\dfrac
	{\displaystyle 
		\overline{Z}(p)[1-\overline{Y}_6(p)]+\overline{Y}_5(p)+\sqrt{\{\overline{Z}(p)[1-\overline{Y}_6(p)]+\overline{Y}_5(p)\}^2+4\overline{Z}(p)\overline{Y}_6(p)[\overline{Z}(p)-\overline{Y}_4(p)]}}
	{2 \overline{Z}(p)\overline{Y}_6(p)}
},
\end{align}
where
\begin{align*}
\overline{Y}_6(p)
\equiv
{1\over n} \sum_{i=1}^{n} T(Y^p_i).
\end{align*}
By plugging \eqref{sigma-estimator} into \eqref{estimator-mu}, the closed-form estimator for $\mu$ is given by
\begin{align}\label{estimator-mu-1}
\widehat{\mu}(p)
=
%
{\color{black}
	\dfrac{\overline{Z}(p)}{\widehat{\sigma}(p)\overline{Y}_5(p)-\overline{Y}_4(p)},
}
\end{align}

%
Furthermore, notice that 
${p}$ is obtained solving the non-linear equation
\begin{align*}
	\log(\widehat{\mu}(p))-\psi^{(0)}(\widehat{\mu}(p))
	=
	\sigma \overline{Y}_6(p)-1
	+	\log\left({1\over\widehat{\sigma}(p)}\right)
	-	
	{1\over n} \sum_{i=1}^{n}\log(T(Y^p_i)).
\end{align*}

Now, return to the distribution in \eqref{pdf-1}. We take $p=1$ in \eqref{sigma-estimator} and \eqref{estimator-mu-1} to obtain the new estimators for $\sigma$ and $\mu$ as
%
\begin{align}\label{mle-sigma-1}
\widehat{\sigma}
=
{\color{black}
\dfrac
{\displaystyle 
	\overline{Z}(1-\overline{Y}_6)+\overline{Y}_5+\sqrt{[\overline{Z}(1-\overline{Y}_6)+\overline{Y}_5]^2+4\overline{Z}\,\overline{Y}_6(\overline{Z}-\overline{Y}_4)}}
{2 \overline{Z}\, \overline{Y}_6}
}
\end{align}
and
\begin{align}\label{estimator-mu-2}
\widehat{\mu}
=
%
{\color{black}
	\dfrac{\overline{Z}}{\widehat{\sigma}\overline{Y}_5-\overline{Y}_4},
}
\end{align}
respectively,
with $\overline{Z}\equiv \overline{Z}(1)$, $\overline{Y}_4\equiv \overline{Y}_4(1)$, $\overline{Y}_5\equiv \overline{Y}_5(1)$ and $\overline{Y}_6\equiv \overline{Y}_6(1)$.

Note that $\widehat{\mu}$ and $\widehat{\sigma}$ in \eqref{estimator-mu-2} and \eqref{mle-sigma-1} are not the ML estimators for $\mu$ and $\sigma$, respectively.

\subsection{Case $a\neq b$}

In this case, as done in \cite{Vila2024}, we use the equation $[{\partial \log(L(\psi,p))/\partial \sigma}]\vert_{p=1}=0$ to obtain the closed-form estimator for $\sigma$:
	\begin{align*}
	\widehat{\sigma}
	=
	{1\over \overline{Y}_6}.
	\end{align*}
	Note that $\widehat{\sigma}$ is the ML estimator for $\sigma$.
	By replacing $\widehat{\sigma}$ into $[{\partial \log(L(\psi,p))/\partial p}]\vert_{p=1}=0$ we have the following closed-form estimator for $\mu$:
\begin{align*}
\widehat{\mu}
=
\dfrac{\displaystyle
	\overline{Y}_6
\big(
1
	+
	\overline{Y}_1
	+
	\overline{Y}_2
	-
	\overline{Y}_4
\big)
}{
	\overline{Y}_5-\overline{Y}_6\overline{Y}_4},
\end{align*}
where  $\overline{Y}_1$,  $\overline{Y}_2$, $\overline{Y}_4$, $\overline{Y}_5$ and $\overline{Y}_6$ were defined in Subsection \ref{Case equal}.

\begin{remark}
Since case $a\neq b$ was studied extensively in \cite{Vila2024}, from now on we will focus our study on case $a=b$.
\end{remark}

\section{Asymptotic behavior of estimators}\label{Asymptotic behavior of estimators}

Let $Y_1,\ldots, Y_n$ be a random sample of size $n$ from the  variable $Y$ with PDF \eqref{pdf-1}. If we further let  $\overline{\boldsymbol{Y}}=(\overline{Y}_1,\overline{Y}_2,\overline{Y}_3,\overline{Y}_4,\overline{Y}_5,\overline{Y}_6)^\top$ and $\boldsymbol{Y}=(Y_1^*,Y_2^*,Y_3^*,Y_4^*,Y_5^*,Y_6^*)^\top$, with $\overline{Y}_i$, $i=1,\ldots,6$, as given in Subsection \ref{Case equal} and 
\begin{align}\label{Ystar}
Y_1^*=\log(Y), \ \
Y_2^*={T''(Y)\over T'(Y)}\, Y Y_1^*, \ \
Y_3^*=	{T'(Y)\over 1+T(Y)}\, YY_1^*, \ \
Y_4^*={T'(Y)\over T(Y)}\,  Y Y_1^*,  \ \
Y_5^*=T'(Y) Y Y_1^*,  \ \ 
Y_6^*=T(Y),
\end{align}
then, by applying strong law of large
numbers, we have
\begin{align*}
	\overline{\boldsymbol{Y}}\stackrel{\rm a.s.}{\longrightarrow}
	\mathbb{E}(\boldsymbol{Y}),
\end{align*}
where ``$\stackrel{\rm a.s.}{\longrightarrow}$'' denotes almost sure convergence.
Hence, continuous-mapping theorem gives,
\begin{align}\label{id-1}
\widehat{\sigma}=g_1(\overline{\boldsymbol{Y}})\stackrel{\rm a.s.}{\longrightarrow}
g_1(\mathbb{E}(\boldsymbol{Y}))
\quad \text{and} \quad
\widehat{\mu}=g_2(\overline{\boldsymbol{Y}})\stackrel{\rm a.s.}{\longrightarrow}
g_2(\mathbb{E}(\boldsymbol{Y})),
\end{align}
with
\begin{multline}\label{def-g1}
g_1(y_1,y_2,y_3,y_4,y_5,y_6)
=
	\dfrac
	{\displaystyle 
		\left(1+\sum_{i=1}^{3}y_i-y_4\right)(1-y_6)
		+y_5
	}
	{2 \displaystyle\left(1+\sum_{i=1}^{3}y_i-y_4\right) y_6}
	\\[0,2cm]
	+
	\dfrac{		\sqrt{\left[\displaystyle\left(1+\sum_{i=1}^{3}y_i-y_4\right)(1-y_6)+y_5\right]^2+4\displaystyle\left(1+\sum_{i=1}^{3}y_i-y_4\right) y_6\left[\displaystyle\left(1+\sum_{i=1}^{3}y_i-y_4\right)-y_4\right]}}{2 \displaystyle\left(1+\sum_{i=1}^{3}y_i-y_4\right) y_6}
\end{multline}
%
and
%
\begin{align}\label{def-g2}
	g_2(y_1,y_2,y_3,y_4,y_5,y_6)
\equiv
		\dfrac{\displaystyle 
		1+\sum_{i=1}^{3}y_i-y_4}
		{g_1(y_1,y_2,y_3,y_4,y_5,y_6)y_5-y_4}.
\end{align}

Furthermore, by Central limit theorem,
\begin{align*}
	\sqrt{n}\big[\overline{\boldsymbol{Y}}-\mathbb{E}(\boldsymbol{Y})\big]\stackrel{\mathscr D}{\longrightarrow} N_6(\bm 0, \bm\Sigma),
\end{align*}
where $\bm\Sigma$ denotes the covariance matrix of $\boldsymbol{Y}$ and ``$\stackrel{\mathscr D}{\longrightarrow}$'' means convergence in distribution.
So, delta method provides
\begin{align}\label{id-2}
	\sqrt{n}
	\left[
	\begin{pmatrix}
	\widehat{\mu}
	\\[0,2cm]
	\widehat{\sigma}
	\end{pmatrix}
	-
	\begin{pmatrix}
	g_2(\mathbb{E}(\bm Y))
	\\[0,2cm]
	g_1(\mathbb{E}(\bm Y))
	\end{pmatrix}
	\right]
	\stackrel{\mathscr D}{\longrightarrow}
	N_2(\bm 0, \bm A\bm\Sigma \bm A^\top),
\end{align}
with $\bm A$ being the partial derivatives matrix defined as
\begin{align*}
	\bm A
	=
	\left. 
	\begin{pmatrix}
	\displaystyle
	{\partial g_2(\bm y)\over\partial y_1} & \displaystyle {\partial g_2(\bm y)\over\partial y_2} & \displaystyle {\partial g_2(\bm y)\over\partial y_3} & \displaystyle {\partial g_2(\bm y)\over\partial y_4} & \displaystyle {\partial g_2(\bm y)\over\partial y_5} & \displaystyle {\partial g_2(\bm y)\over\partial y_6}
	\\[0,5cm]
	\displaystyle
	{\partial g_1(\bm y)\over\partial y_1} & \displaystyle {\partial g_1(\bm y)\over\partial y_2} & \displaystyle {\partial g_1(\bm y)\over\partial y_3} & \displaystyle {\partial g_1(\bm y)\over\partial y_4} & \displaystyle {\partial g_1(\bm y)\over\partial y_5} & \displaystyle {\partial g_1(\bm y)\over\partial y_6}
	\end{pmatrix}\,
	\right\vert_{\bm y=\mathbb{E}(\bm Y)}.
\end{align*}
For simplicity of presentation, we do not present the partial derivatives of $g_j$, $j=1,2$, here.

Note that, in \eqref{id-1} and \eqref{id-2} it is extremely complicated to analytically calculate $g_j(\mathbb{E}(\bm Y))$, $j=1,2$, for general generators $T$, so a numerical analysis for this calculation is required. In the case that $g_1(\mathbb{E}(\bm Y))=\sigma$ and $g_2(\mathbb{E}(\bm Y))=\mu$, from \eqref{id-1} we obtain the strong consistency of the estimators $\widehat{\sigma}$ and $\widehat{\mu}$. The following result shows that for generators of type $T(x)=x^{-s}$, for some $s\ne 0$, the strong consistency property of the estimators $\widehat{\sigma}$ and $\widehat{\mu}$ is satisfied.

\begin{theorem}
	If $T(x)=x^{-s}$, for some $s\ne 0$, then
	$g_1(\mathbb{E}(\bm Y))=\sigma$ and $g_2(\mathbb{E}(\bm Y))=\mu$.
\end{theorem}
\begin{proof}
Note that the stochastic representation in \eqref{rep-st-2} can be expressed as
\begin{align*}
	Y\stackrel{\mathscr{D}}{=}(1-B)Z_1^{-1/s}+BZ_2^{-1/s},
\end{align*}
where $B\sim {\rm Bernoulli}(1/(\sigma+1))$  is independent of $Z_1$ and $Z_2$, such that $Z_j\sim{\rm Gamma}(\mu+j-1,1/(\mu\sigma))$, $j=1,2$.
By using the identities 
$T'(x)=-s x^{-s-1}$ and
$T''(x)=s(s+1) x^{-s-2}$,
and the above stochastic representation,
the statistics $Y_k^*$, $k=1,\ldots,6$, given in \eqref{Ystar}, can be written as
\begin{align*}
&Y_1^*
\stackrel{\mathscr{D}}{=}
-{1\over s}[(1-B)\log(Z_1)+B\log(Z_2)],
\\[0,2cm]
&Y_2^*
\stackrel{\mathscr{D}}{=}
-(s+1)Y_1^*,
\\[0,2cm]
&Y_3^*
\stackrel{\mathscr{D}}{=}
{(1-B)Z_1\log(Z_1)+BZ_2\log(Z_2)\over 1+(1-B)Z_1+BZ_2},
\\[0,2cm]
&Y_4^*
\stackrel{\mathscr{D}}{=}
-sY_1^*,
\\[0,2cm]
&Y_5^*
\stackrel{\mathscr{D}}{=}
(1-B)Z_1\log(Z_1)+BZ_2\log(Z_2),
\\[0,2cm]
&Y_6^*
\stackrel{\mathscr{D}}{=}
(1-B)Z_1+BZ_2.
\end{align*}
Consequently, by employing the well-known identities $\mathbb{E}[\log(Z_j)]=\psi^{(0)}(\mu+j-1)-\log(\mu\sigma)$, $\mathbb{E}[Z_j\log(Z_j)]=[(\mu+j-1)/(\mu\sigma)][\psi^{(0)}(\mu+j)-\log(\mu\sigma)]$, $j=1,2$, and $\psi^{(0)}(z+1)=\psi^{(0)}(z)+1/z$, we get
\begin{align*}
&\mathbb{E}(Y_1^*)
=
-{1\over s}
\left[
\psi^{(0)}(\mu)
-\log(\mu\sigma)
+
{1\over \mu(\sigma+1)}
\right],
\\[0,2cm]
&\mathbb{E}(Y_2^*)
=
{s+1\over s}
\left[
\psi^{(0)}(\mu)
-\log(\mu\sigma)
+
{1\over \mu(\sigma+1)}
\right],
\\[0,2cm]
&\mathbb{E}(Y_3^*)
=
{1\over\sigma +1}
\left[
\psi^{(0)}(\mu)
-
\log(\mu\sigma)
+
{1\over\mu}
\right],
\\[0,2cm]
&\mathbb{E}(Y_4^*)
=
\psi^{(0)}(\mu)
-\log(\mu\sigma)
+
{1\over \mu(\sigma+1)},
\\[0,2cm]
&\mathbb{E}(Y_5^*)
=
{1\over\sigma+1}\left(1+{\mu+1\over \mu\sigma}\right) \left[\psi^{(0)}(\mu)-\log(\mu\sigma)
+{1\over\mu}\right]
+
{1\over\mu\sigma(\sigma+1)},
\\[0,2cm]
&\mathbb{E}(Y_6^*)
=
{1\over\sigma+1}\left(1+{\mu+1\over\mu\sigma}\right).
\end{align*}
Hence, by using the definition of $g_1$ and $g_2$, given in \eqref{def-g1} and \eqref{def-g2}, respectively, a simple observation shows that
\begin{align*}
g_1(\mathbb{E}(\bm Y))
=
g_1((\mathbb{E}(Y_1^*),\mathbb{E}(Y_2^*),\mathbb{E}(Y_3^*),\mathbb{E}(Y_4^*),\mathbb{E}(Y_5^*),\mathbb{E}(Y_6^*))^\top)
=\sigma
\end{align*}
and
\begin{align*}
g_2(\mathbb{E}(\bm Y))
=
g_2((\mathbb{E}(Y_1^*),\mathbb{E}(Y_2^*),\mathbb{E}(Y_3^*),\mathbb{E}(Y_4^*),\mathbb{E}(Y_5^*),\mathbb{E}(Y_6^*))^\top)
=\mu.
\end{align*}
This completes the proof.
\end{proof}

\section{Closed-form estimators of Table \ref{table:1-0}}\label{Sec:5}

In this section, we provide closed-form estimators for the paramaters that index the distributions of Table \ref{table:1-0}. We use the  notation $\widehat{\sigma}$ stated in \eqref{mle-sigma-1}, which for ease of reading we make available again in this section:
\begin{align}\label{mle-sigma-1-4}
\widehat{\sigma}
=
%
{\color{black}
	\dfrac
	{\displaystyle
		\overline{Z}(1-\overline{Y}_6)+\overline{Y}_5+\sqrt{[\overline{Z}(1-\overline{Y}_6)+\overline{Y}_5]^2+4\overline{Z}\,\overline{Y}_6(\overline{Z}-\overline{Y}_4)}}
	{2 \overline{Z}\, \overline{Y}_6},
}
\end{align}
where $\overline{Z},\overline{Y}_4,\overline{Y}_5$ and $\overline{Y}_6$ were defined in Subsection \ref{Case equal}.

\paragraph{Weighted Lindley distribution.}
	By considering the parameters $\mu=\phi$, $\sigma=\lambda/\phi$ and generator $T(x)=x$ of the weighted Lindley distribution, given in Table \ref{table:1-0}, from \eqref{mle-sigma-1} and \eqref{estimator-mu-2} the closed-form  estimators for  $\phi$ and $\lambda$ are given by
	\begin{align}\label{estimator-mu-2-2}
	\widehat{\phi}
	=
	{\color{black}
		\dfrac{\overline{Z}}{\widehat{\sigma}\overline{Y}_5-\overline{Y}_4}
	}
	\end{align}
	and
	\begin{align}\label{mle-sigma-1-1}
	\widehat{\lambda}
	=
	{\color{black}
		\dfrac{\widehat{\sigma}\overline{Z}}{\widehat{\sigma}\overline{Y}_5-\overline{Y}_4}
	},
	\end{align}
%
	respectively. Here, $\widehat{\sigma}$ is as given in \eqref{mle-sigma-1-4}, with
	{\color{black}
	\begin{align*}
	&\overline{Z}
	=
	1
	+
	{1\over n}
	\sum_{i=1}^{n}
	{Y_i\log(Y_i)\over 1+Y_i},
	\\[0,2cm]
	&\overline{Y}_4
	=
	{1\over n}\sum_{i=1}^{n} \log(Y_i),
	\\[0,2cm]
	&\overline{Y}_5
	=
	{1\over n} \sum_{i=1}^{n} Y_i \log(Y_i),
		\\[0,2cm]
	&\overline{Y}_6
	=
	{1\over n} \sum_{i=1}^{n} Y_i.
	\end{align*}
}

	We emphasize that estimators $\widehat{\lambda}$ and $\widehat{\phi}$ for the weighted Lindley distribution have appeared in reference \cite{Kim2021}.

\paragraph{Weighted inverse Lindley distribution.}
	By considering the parameters  $\mu=\phi$, $\sigma=\lambda/\phi$ and generator $T(x)=1/x$ of the weighted inverse Lindley distribution, given in Table \ref{table:1-0}, from \eqref{mle-sigma-1} and \eqref{estimator-mu-2} the closed-form  estimators for $\lambda$ and $\phi$ have the same form of the estimators given in \eqref{mle-sigma-1-1} and \eqref{estimator-mu-2-2}, where
		\begin{align*}
	&\overline{Z}
	=
	1
	-
	{1\over n}
	\sum_{i=1}^{n}
	{Y_i^{-1}\log(Y_i)\over 1+Y_i^{-1}},
	\\[0,2cm]
	&\overline{Y}_4
	=
	-{1\over n}\sum_{i=1}^{n} \log(Y_i),
	\\[0,2cm]
	&\overline{Y}_5
	=
	-
	{1\over n} \sum_{i=1}^{n} Y_i^{-1} \log(Y_i),
	\\[0,2cm]
	&\overline{Y}_6
	=
	{1\over n} \sum_{i=1}^{n} Y_i^{-1}.
	\end{align*}

\paragraph{New weighted exponentiated Lindley distribution.}
	By considering the parameters  $\mu=\phi$, $\sigma=\lambda/\phi$ and generator $T(x)=\log(x+1)$ of new weighted exponentiated Lindley distribution, given  in Table \ref{table:1-0}, from \eqref{mle-sigma-1} and \eqref{estimator-mu-2} the closed-form  estimators for $\lambda$ and $\phi$ have the same form of the estimators given in  \eqref{mle-sigma-1-1} and \eqref{estimator-mu-2-2}, where
		\begin{align*}
	&\overline{Z}
	=
	1
	+
	{1\over n}
	\sum_{i=1}^{n}
	\log(Y_i)
	-
	{1\over n}
	\sum_{i=1}^{n}
	{Y_i \log(Y_i)\over Y_i+1}
	+
	{1\over n}
	\sum_{i=1}^{n}
	{Y_i\log(Y_i)\over (Y_i+1)(1+\log(Y_i+1))}
	-
	\overline{Y}_4,
	\\[0,2cm]
	&\overline{Y}_4
	=
	{1\over n}\sum_{i=1}^{n} {Y_i \log(Y_i)\over (Y_i+1)\log(Y_i+1)},
	\\[0,2cm]
	&\overline{Y}_5
	=
	{1\over n} \sum_{i=1}^{n} {Y_i \log(Y_i)\over Y_i+1},
	\\[0,2cm]
	&\overline{Y}_6
	=
	{1\over n} \sum_{i=1}^{n} \log(Y_i+1).
	\end{align*}

\paragraph{New weighted log-Lindley distribution.}
	By considering the parameters  $\mu=\phi$, $\sigma=\lambda/\phi$ and generator $T(x)=\exp(x)-1$ of new weighted log-Lindley distribution, given  in Table \ref{table:1-0}, from \eqref{mle-sigma-1} and \eqref{estimator-mu-2} the closed-form  estimators for $\lambda$ and $\phi$ have the same form of the estimators given in  \eqref{mle-sigma-1-1} and \eqref{estimator-mu-2-2}, where
		\begin{align*}
	&\overline{Z}
	=
	1
	+
	{1\over n}
	\sum_{i=1}^{n}
	\log(Y_i)
	+
	{2\over n}
	\sum_{i=1}^{n}
	Y_i \log(Y_i)
	-
		\overline{Y}_4,
	\\[0,2cm]
	&\overline{Y}_4
	=
	{1\over n}\sum_{i=1}^{n} {\exp(Y_i)Y_i \log(Y_i)\over \exp(Y_i) -1},
	\\[0,2cm]
	&\overline{Y}_5
	=
	{1\over n} \sum_{i=1}^{n} \exp(Y_i) Y_i \log(Y_i),
	\\[0,2cm]
	&\overline{Y}_6
	=
	{1\over n} \sum_{i=1}^{n} \exp(Y_i)-1.
	\end{align*}

\paragraph{Weighted Nakagami distribution.}
By considering the parameters  $\mu=m$, $\sigma=1/\Omega$ and generator $T(x)=x^2$ of new weighted Nakagami distribution, given  in Table \ref{table:1-0}, from \eqref{mle-sigma-1} and \eqref{estimator-mu-2} the closed-form  estimators for $\Omega$ and $m$ are given by
\begin{align}\label{mle-sigma-1-1-1}
\widehat{\Omega}
=
{\color{black}
	\dfrac
	{2 \overline{Z}\, \overline{Y}_6}
	{\displaystyle
		\overline{Z}(1-\overline{Y}_6)+\overline{Y}_5+\sqrt{[\overline{Z}(1-\overline{Y}_6)+\overline{Y}_5]^2+4\overline{Z}\,\overline{Y}_6(\overline{Z}-\overline{Y}_4)}},
}
\end{align}
%
and
%
\begin{align}\label{estimator-mu-2-2-2}
\widehat{m}
=
%
{\color{black}
	\dfrac{\widehat{\Omega}\,\overline{Z}}{\overline{Y}_5-\widehat{\Omega}\,\overline{Y}_4},
}
\end{align}
respectively, where
\begin{align*}
&\overline{Z}
=
1
+
{2\over n}
\sum_{i=1}^{n}
{Y_i^2 \log(Y_i)\over 1+Y_i^2},
\\[0,2cm]
&\overline{Y}_4
=
{2\over n}\sum_{i=1}^{n} \log(Y_i),
\\[0,2cm]
&\overline{Y}_5
=
{2\over n} \sum_{i=1}^{n} Y_i^2 \log(Y_i),
\\[0,2cm]
&
\overline{Y}_6
=
{1\over n} \sum_{i=1}^{n} Y_i^2.
\end{align*}

\paragraph{Weighted inverse Nakagami distribution.}
	By considering the parameters  $\mu=m$, $\sigma=1/\Omega$ and generator $T(x)=1/x^2$ of weighted inverse Nakagami distribution, given  in Table \ref{table:1-0}, from \eqref{mle-sigma-1} and \eqref{estimator-mu-2} the closed-form  estimators for $\Omega$ and $m$ have the same form of the estimators given in  \eqref{mle-sigma-1-1-1} and \eqref{estimator-mu-2-2-2}, where
	\begin{align*}
	&\overline{Z}
	=
	1
	-
	{2\over n}
	\sum_{i=1}^{n}
	{Y_i^{-2}\log(Y_i)\over 1+Y_i^{-2}},
	\\[0,2cm]
	&\overline{Y}_4
	=
	-{2\over n}\sum_{i=1}^{n} \log(Y_i),
	\\[0,2cm]
	&\overline{Y}_5
	=
	-{2\over n} \sum_{i=1}^{n} Y_i^{-2}\, \log(Y_i),
	\\[0,2cm]
	&
	\overline{Y}_6
	=
	{1\over n} \sum_{i=1}^{n} Y_i^{-2}.
	\end{align*}

\paragraph{New weighted exponentiated Nakagami distribution.}
	By considering the parameters  $\mu=m$, $\sigma=1/\Omega$ and generator $T(x)=\log(x^2+1)$ of new weighted exponentiated Nakagami distribution, given  in Table \ref{table:1-0}, from \eqref{mle-sigma-1} and \eqref{estimator-mu-2} the closed-form  estimators for $\Omega$ and $m$ have the same form of the estimators given in  \eqref{mle-sigma-1-1-1} and \eqref{estimator-mu-2-2-2}, where
	\begin{align*}
	&\overline{Z}
	=
	1
	+
	{1\over n}
	\sum_{i=1}^{n}
	\log(Y_i)
	-
	{1\over n}
	\sum_{i=1}^{n}
	{(Y_i^2-1)\log(Y_i)\over Y_i^2+1}
	+
	{2\over n}
	\sum_{i=1}^{n}
	{Y_i^2\log(Y_i)\over (Y_i^2+1)(1+\log(Y_i^2+1))}
	-
		\overline{Y}_4,
	\\[0,2cm]
	&\overline{Y}_4
	=
	{2\over n}\sum_{i=1}^{n} {Y_i^2\log(Y_i)\over (Y_i^2+1)\log(Y_i^2+1)},
	\\[0,2cm]
	&\overline{Y}_5
	=
	{2\over n} \sum_{i=1}^{n} {Y_i^2\log(Y_i)\over Y_i^2+1},
	\\[0,2cm]
	&
	\overline{Y}_6
	=
	{1\over n} \sum_{i=1}^{n} \log(Y_i^2+1).
	\end{align*}

\paragraph{New weighted log-Nakagami distribution.}
	By considering the parameters  $\mu=m$, $\sigma=1/\Omega$ and generator $T(x)=\exp(x^2)-1$ of new weighted log-Nakagami distribution, given  in Table \ref{table:1-0}, from \eqref{mle-sigma-1} and \eqref{estimator-mu-2} the closed-form  estimators for $\Omega$ and $m$ have the same form of the estimators given in  \eqref{mle-sigma-1-1-1} and \eqref{estimator-mu-2-2-2}, where
\begin{align*}
&\overline{Z}
=
1
+
{2\over n}
\sum_{i=1}^{n}
\log(Y_i)
+
{4\over n}
\sum_{i=1}^{n}
{Y_i^2\log(Y_i)}
-
\overline{Y}_4,
\\[0,2cm]
&\overline{Y}_4
=
{2\over n}\sum_{i=1}^{n} {\exp(Y_i^2) Y_i^2\log(Y_i)\over \exp(Y_i^2)-1},
\\[0,2cm]
&\overline{Y}_5
=
{2\over n} \sum_{i=1}^{n} \exp(Y_i^2) Y_i^2 \log(Y_i),
\\[0,2cm]
&
\overline{Y}_6
=
{1\over n} \sum_{i=1}^{n} \exp(Y_i^2)-1.
\end{align*}

\section{Simulation study}\label{Sec:6}

In this section, we carry out a Monte Carlo simulation study for evaluating the performance of the proposed estimators. For illustrative purposes, we only present the results for the weighted Lindley distribution  (the results for other distributions are
similar and then ommited here). As the estimators in \ref{Sec:5} are biased \citep{RLR2016}, we then propose a bootstrap bias-reduced version of these proposed ML estimators as
\begin{equation*}\label{asy:08}
\widehat{\theta}^{*} = 2\widehat{\theta} - \frac{1}{B}\sum_{b=1}^{B}\widehat{\theta}^{(b)},
\end{equation*}
where $\widehat{\theta}\in\{\widehat{\Omega},
\widehat{m},
\widehat{\beta},
\widehat{\alpha},
\widehat{\tau}^{\,2},
\widehat{\nu}\}$,
$B$ is the number of bootstrap replications, and $\widehat{\theta}^{(b)}$ is the $b$-th bootstrap replicate from the $b$-th bootstrap sample. We study the performance of the proposed bootstrap bias-reduced estimators by computing the relative bias (RB) and root mean square error (RMSE), as
\begin{eqnarray*}
 \widehat{\textrm{RB}}(\widehat{\theta}^{*}) =  \left|\frac{\frac{1}{N} \sum_{i = 1}^{N} \widehat{\theta}^{*(i)} - \theta}{\theta}\right| ,  \quad
\widehat{\mathrm{RMSE}}(\widehat{\theta}^{*}) = {\sqrt{\frac{1}{N} \sum_{i = 1}^{N} (\widehat{\theta}^{*(i)} - \theta)^2}},
\end{eqnarray*}
where $\theta$ and $\widehat{\theta}^{*(i)}$ are the true parameter value and its $i$-th bootstrap bias-reduced estimate, and $N$ is the number of Monte Carlo replications.

The simulation scenario considers the following setting: $n \in \{20,50,100,200,400,1000\}$; $\phi\in \{0.5,1,3,5,9\}$, and $\lambda=1$, with $N=1,000$ and $B=200$ Monte Carlo and bootstrap replications, respectively, for each sample size. The
\texttt{R} software was used to do all numerical calculations; see \url{http://cran.r-project.org}.

The ML estimation results for the considered weighted Lindley distribution are presented in Figures \ref{fig_dagum_MC1} and \ref{fig_dagum_MC2}. We observe that both RBs and RMSEs approach zero as $n$ grows. Figure~\ref{fig_dagum_MC3} presents the execution times of the Monte Carlo simulations for the sample sizes considered. From Figure~\ref{fig_dagum_MC3}, we observe that on average, the 1,000 replicates do not take more than 24 seconds to produce the results. In short, the proposed bootstrap bias-reduced ML estimators may be good options for estimating the parameters associated with the distributions studied in this manuscript.

\begin{figure}[H]
\vspace{-0.25cm}
\centering
{\includegraphics[height=5.5cm,width=5.5cm]{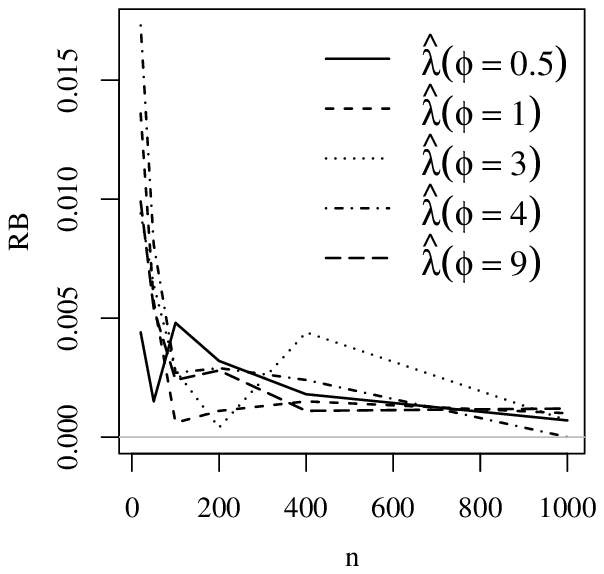}}\hspace{-0.25cm}
{\includegraphics[height=5.5cm,width=5.5cm]{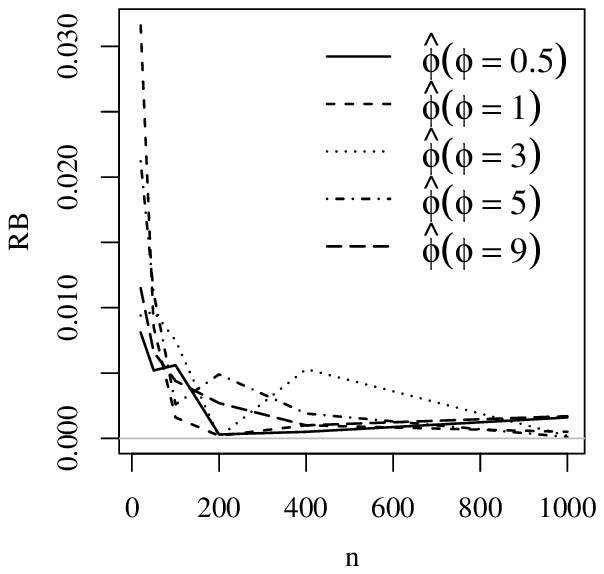}}\hspace{-0.25cm}
\vspace{-0.2cm}
\caption{Empirical RB of the ML estimators for the weighted Lindley distribution.}
\label{fig_dagum_MC1}
\end{figure}

\begin{figure}[H]
\vspace{-0.25cm}
\centering
{\includegraphics[height=5.5cm,width=5.5cm]{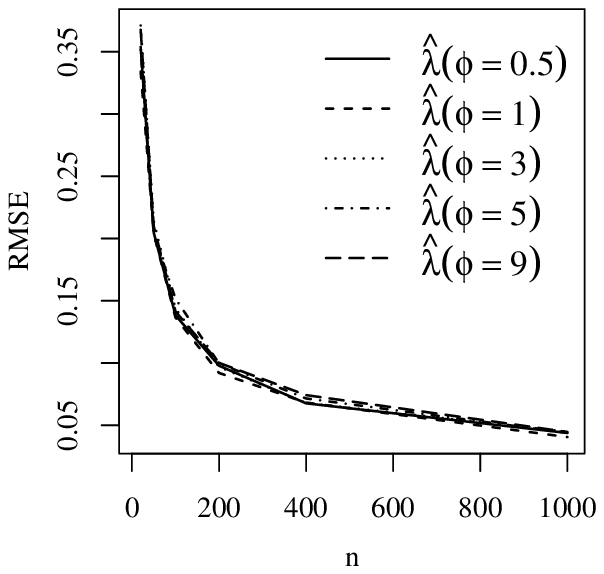}}\hspace{-0.25cm}
{\includegraphics[height=5.5cm,width=5.5cm]{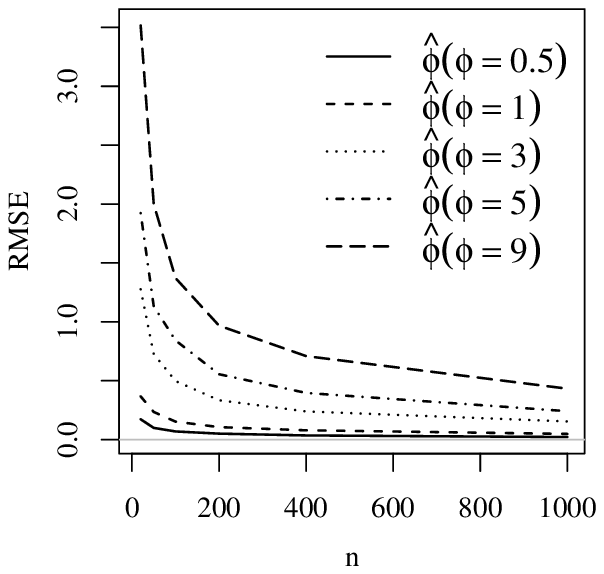}}\hspace{-0.25cm}

%
\vspace{-0.2cm}
\caption{Empirical RMSE of the ML estimators for the weighted Lindley distribution.}
\label{fig_dagum_MC2}
\end{figure}

\begin{figure}[H]
\vspace{-0.25cm}
\centering
{\includegraphics[height=5.5cm,width=5.5cm]{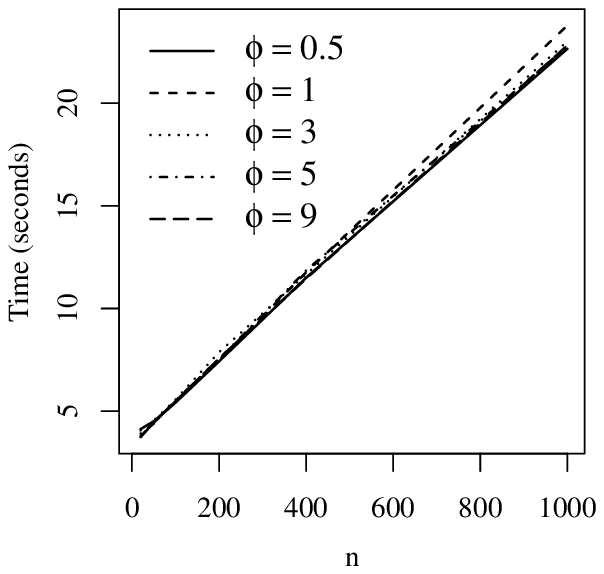}}\hspace{-0.25cm}
%
\vspace{-0.2cm}
\caption{Empirical time of the ML estimators for the weighted Lindley distribution.}
\label{fig_dagum_MC3}
\end{figure}

	\paragraph*{Acknowledgements}
 This study was financed in part by the
Coordenação de Aperfeiçoamento de Pessoal de Nível Superior - Brasil (CAPES) - Finance Code 001.
	
	\paragraph*{Disclosure statement}
	There are no conflicts of interest to disclose.

\end{document}